\documentclass[12pt,reqno,tbtags,spconf]{amsproc}
\usepackage{graphicx,amscd}
\usepackage{calrsfs}
\usepackage{amssymb}
\usepackage{amsmath}
\usepackage[unicode,bookmarks,bookmarksopen,bookmarksopenlevel=2,colorlinks,linkcolor=blue,citecolor=green]{hyperref}

\newtheorem{theorem}{Theorem}[section]
\newtheorem{prop}[theorem]{Proposition}
\newtheorem{lemma}[theorem]{Lemma}
\newtheorem{cor}[theorem]{Corollary}
\theoremstyle{definition}
\newtheorem{exam}[theorem]{Example}
\newtheorem{remark}[theorem]{Remark}
\newtheorem{defn}[theorem]{Definition}

\numberwithin{equation}{section}
\newcommand{\subhead}[1]{ \hspace{1em}{\bf #1} \newline}

\def\pa{\partial}
\def\op{\operatorname}
\def\ft{\footnote}
\def\wt{\widetilde}
\def\wh{\widehat}
\def\Bibitem{\bibitem}

\newcommand{\jour}[1]{{\it #1}, }
\newcommand{\yr}[1]{(#1) }
\newcommand{\pages}[1]{p. #1. }
\newcommand{\paper}[1]{``#1'', }
\newcommand{\by}[1]{{#1}, }
\newcommand{\publ}[1]{#1, }
\def\book{}
\newcommand{\issue}[1]{:#1, }
\newcommand{\vol}[1]{{\bf #1} }
\newcommand{\inbook}{ {\it in the book} }
\newcommand{\transl}{}
\def\nofrills{}
\def\publaddr{}
\def\serial{}

\begin{document}

\title[New Class of Solutions to the WDVV Associativity Equations]{Linearly Degenerate Hamiltonian PDEs and a New Class of Solutions to the WDVV Associativity Equations}

\author{B.~A.~Dubrovin}
\address{B.~A.~Dubrovin, Scuola Internazionale Superiore di Studi Avanzati, Trieste;
Moscow State University,
Laboratory of Geometric methods in Mathematical Physics;
V. A.~Steklov Mathematical Institute, Moscow}
\email{dubrovin@sissa.it}
\thanks{This work is partially supported by the European Research
Council Advanced Grant FroM-PDE, by the Russian Federation Government Grant No. 2010-220-01-077
(ag. \#11.G34.31.0005), and by PRIN 2008 Grant ``Geometric methods in the theory of nonlinear waves
and their applications'' of Italian Ministry of Universities and Researches.}

\author{M.~V.~Pavlov}
\address{M.~V.~Pavlov, Moscow State University,
Laboratory of Geometric methods in Mathematical Physics;\\
P. N.~Lebedev Physical Institute of RAS}
\email{maxim@math.sinica.edu.tw}
\thanks{M.~V.~Pavlov is partially supported by the grant of Presidium
of RAS \textquotedblleft Fundamental Problems of Nonlinear Dynamics\textquotedblright\ and by
RFBR grant 11-01-00197. M.~V.~Pavlov and S.~A.~Zykov are
grateful to the SISSA in Trieste (Italy), where a part of this work has been
done.}

\author{S.~A.~Zykov}
\address{S.~A.~Zykov, University of Salento, Lecce,Department of Physics;\\
Institute of Metal Physics, Ural branch of RAS, Ekaterinburg}
\email{zykov@le.infn.it}
\thanks{S.~A.~Zykov is also partially supported by the INFN Section in Lecce, project No.~LE41.}


\date{May 30, 2011}

\begin{abstract}
We define a new class of solutions to the WDVV associativity equations. This class is
determined by the property that one of the commuting PDEs associated with such a WDVV solution
is linearly degenerate. We reduce the problem of classifying such
solutions of the WDVV equations to the
particular case of the so-called algebraic Riccati equation and, in this way, arrive
at a complete classification of irreducible solutions. \end{abstract}

\maketitle

\rightline{{\it To the memory of V. I. Arnold}}

\tableofcontents

\section{Introduction}

The Witten--Dijkgraaf--E.~Verlinde--H.~Verlinde (WDVV) system of
associativity
equations is the
overdetermined system of partial differential equations
\begin{equation}\label{wdvv}
\frac{\pa^3 F}{\pa v^\alpha\pa v^\beta \pa v^\lambda}\,\eta^{\lambda\mu}\frac{\pa^3 F}{\pa
v^\mu\pa v^\gamma \pa v^\delta}=\frac{\pa^3 F}{\pa v^\delta\pa v^\beta \pa v^\lambda}\,
\eta^{\lambda\mu}\frac{\pa^3 F}{\pa v^\mu\pa v^\gamma \pa v^\alpha}, \qquad \alpha, \beta, \gamma,
\delta=1, \dots, n,
\end{equation}
for a function $F=F({\bf v})$, ${\bf v}=(v^1,\dots,v^n)$, satisfying the conditions
$$
\frac{\pa^3 F}{\pa v^\alpha\pa v^\beta \pa v^1} =\eta_{\alpha\,\beta}.
$$
Here $(\eta_{\alpha\beta})_{1\le\alpha,\beta\le n}$ and $(\eta^{\alpha\beta})_{1\le\alpha,\beta\le
n}$ are mutually inverse constant symmetric nonsingular matrices, that is,
$\eta_{\alpha\lambda}\eta^{\lambda\beta}=\delta_\alpha^\beta$.
Throughout this section summation over repeated Greek indices will be assumed.

Recall \cite{monte} that the solutions to the WDVV associativity equations are in one-to-one
correspondence with the $n$-parameter families
of $n$-dimensional commutative associative algebras
$$
\mathcal{A}_{\bf v}=\op{span}(e_1, \dots, e_n)
$$
with a unit $e=e_1$ equipped with a symmetric nondegenerate invariant bilinear form
$(\kern6pt,\kern5pt)$ such that the structure constants are expressed via the third
derivatives of a
function $F$, called the \emph{potential}:
\begin{align*}
e_\alpha \cdot e_\beta&=c_{\alpha\beta}^\gamma({\bf v}) e_\gamma, \qquad \alpha, \beta=1, \dots, n,\\
e_1\cdot e_\alpha&=e_\alpha\quad \text{for any }\,\alpha,\\
(e_\alpha, e_\beta)&=\eta_{\alpha\beta},\\
(e_\alpha \cdot e_\beta, e_\gamma)&= (e_\alpha, e_\beta\cdot e_\gamma)=\eta_{\gamma\lambda}
c_{\alpha\beta}^\lambda({\bf v})=\smash[t]{\frac{\pa^3 F({\bf v})}
{\pa v^\alpha \pa v^\beta \pa v^\gamma}}\,.
\end{align*}

If, in addition, the function $F$ satisfies a certain quasi-homogeneity
condition, then one arrives at
a local description of \emph{Frobenius manifolds} (see details in \cite{monte}).
On these manifolds the natural metric
\begin{equation}\label{metr}
ds^2 =\eta_{\alpha\beta}\,dv^\alpha\,dv^\beta
\end{equation}
(not necessarily positive definite) is defined.
The variables $v^1$, \dots, $v^n$
are \emph{flat coordinates} for this metric. The algebra ${\mathcal{A}}_{\bf v}$ is identified with
the tangent space to the manifold at the point ${\bf v}$:
$$
e_\alpha\leftrightarrow \frac{\pa}{\pa v^\alpha}\,;
$$
see \cite{monte} for more details about the coordinate-free geometric description of Frobenius
manifolds.

A solution to the associativity equations \eqref{wdvv} is called \emph{semisimple} if the algebra
${\mathcal{A}}_{\bf v}$ has no nilpotent elements for a generic point ${\bf v}$. It was
proved in \cite{npb} that, in the semisimple
case, there exist
local \emph{canonical coordinates} $u_i=u_i({\bf v})$,
$i=1, \dots, n$, such that the multiplication table takes the
standard form
$$
\frac{\pa }{\pa u_i} \cdot \frac{\pa }{\pa u_j}=\delta_{ij}\,\frac{\pa }{\pa u_i}\,.
$$
The metric \eqref{metr} becomes diagonal in these canonical coordinates:
$$
ds^2 =\sum_{i=1}^n h_i^2({\bf u})\,du_i^2.
$$
Moreover, this is a \emph{Egorov} metric (see \cite{egor}), which means that
the \emph{rotation coefficients}
\begin{equation}\label{rot1}
\gamma_{ij}({\bf u})=\frac1{h_j}\,\frac{\pa h_i}{\pa u_j}
\end{equation}
are symmetric in $i$ and $j$, i.e., $\gamma_{ji}=\gamma_{ij}$.
They satisfy the following system of
\emph{Darboux--Egorov equations} \cite{darboux}:
\begin{alignat}2
\label{egor1} \frac{\pa \gamma_{ij}}{\pa u_k}&=\gamma_{ik} \gamma_{kj}&&\qquad
\text{for distinct }i,j,k,\\
\sum_{k=1}^n \frac{\pa \gamma_{ij}}{\pa u_k}&=0&&\qquad \text{for } i\ne j. \label{egor11}
\end{alignat}
Any solution to the Darboux--Egorov equations comes from a semisimple solution to the WDVV
associativity equations. The reconstruction procedure of the latter involves solutions to the
following system of linear differential equations for a vector-function $\psi=(\psi_1({\bf
u}),\dots,\psi_n({\bf u}))$:
\begin{align}\label{egor2}
\frac{\pa \psi_i}{\pa u_j}&=\gamma_{ij} \psi_j, \qquad i\neq j,\\
\sum_{k=1}^n \frac{\pa \psi}{\pa u_k}&=0. \label{egor3}
\end{align}
Let $\psi_{i\alpha}=\psi_{i\alpha}({\bf u})$, $\alpha=1, \dots, n$, be
a system of $n$ linearly
independent solutions to system \eqref{egor2}, \eqref{egor3}.
The reconstruction depends on a choice of one
of these solutions to be identified with the Lam\'e coefficients of the invariant metric
\eqref{metr}; suppose that the chosen solution corresponds to
$\alpha=1$, that is, $h_i=\psi_{i1}$. Then
\begin{align*}
\eta_{\alpha\beta}&=\sum_{i=1}^n \psi_{i\alpha}\psi_{i\beta},\\
dv_\alpha&= \sum_{i=1}^n \psi_{i\alpha}\psi_{i1}\,du_i,\\
\frac{\pa^3 F}{\pa v^\alpha\pa v^\beta\pa v^\gamma}&=\sum_{i=1}^n
\frac{\psi_{i\alpha}\psi_{i\beta}\psi_{i\gamma}}{\psi_{i1}}\,.
\end{align*}
We also mention the following formula
for the differentials of the second derivatives
\begin{equation}\label{egor41}
\Omega_{\alpha\beta}=
\frac{\pa^2 F}{\pa v^\alpha \pa v^\beta}
\end{equation}
of the potential $F$:
\begin{equation}\label{egor4}
d\Omega_{\alpha\beta}=\sum_{i=1}^n \psi_{i\alpha} \psi_{i\beta}\,du_i.
\end{equation}

As shown in \cite{funcan}, the Darboux--Egorov
system \eqref{egor1}--\eqref{egor11} can be
identified with a special reduction of the $n$-wave
system
well known in the theory of integrable PDEs and written in the form suggested in \cite{77}. It can also be
embedded in the framework of the
$n$KP system (see, e.g., \cite{leur}). All known particular solutions
to the associativity equations
correspond to further reductions of the $n$-wave system
to a system of ODEs. For example, the
semisimple Frobenius manifolds are determined by the homogeneity condition on the
rotation coefficients, or
the \emph{scaling} reduction
$$
\sum_{k=1}^n u_k\,\frac{\pa \gamma_{ij}}{\pa u_k} =-\gamma_{ij}, \qquad i\neq j.
$$
This condition corresponds to the quasi-homogeneity axiom of the
theory of Frobenius manifolds (see \cite{npb} and
\cite{monte}). Other particular classes of solutions (such as solitons,
algebro-geometric solutions, and
degenerate Frobenius manifolds) also naturally arise in the framework of the $n$-wave system.

In this paper we introduce another class of solutions to the WDVV
equations. Before describing this class, we
recall the connection between the associativity equations
and integrable hierarchies. Let
$\theta=\theta({\bf v})$ be a solution to the system
of linear differential equations
\begin{equation}\label{hier1}
\frac{\pa^2 \theta}{\pa v^\alpha\pa v^\beta}=c_{\alpha\beta}^\gamma\,\frac{\pa^2 \theta}{\pa v^1 \pa
v^\gamma}, \qquad \alpha, \beta=1, \dots, n.
\end{equation}
Consider the following system of first-order quasilinear PDEs for the
vector-function ${\bf v}={\bf v}(x,t)$:
\begin{equation}\label{hier2}
{\bf v}_t = [\nabla \theta({\bf v})]_x.
\end{equation}
This is a Hamiltonian PDE with Hamiltonian
$H=\int \theta({\bf v})\, dx$ and Poisson bracket $\{ v^\alpha(x), v^\beta(y)\}
=\eta^{\alpha\beta} \delta'(x-y)$ (see \cite{dn}). All Hamiltonian systems of the
form \eqref{hier1}, \eqref{hier2} pairwise commute. Moreover,
Hamiltonians \eqref{hier1} satisfy certain \emph{completeness} conditions (see \cite{tsarev}).
Thus, any such system \eqref{hier2} can be considered as a
completely integrable Hamiltonian system of PDEs.

In the semisimple case all such PDEs diagonalize in the canonical coordinates, i.e.,
\begin{equation}\label{hier3}
{\bf u}_t =\Lambda({\bf u}) {\bf u}_x, \qquad \Lambda({\bf u}) =\op{diag} (\lambda_1({\bf u}),
\dots, \lambda_n({\bf u})).
\end{equation}
Thus, the canonical coordinates are \emph{Riemann invariants} for the quasilinear systems
\eqref{hier2}. For a generic solution to \eqref{hier1}, the characteristic
velocities are pairwise
distinct, i.e.,
\begin{equation}\label{hier4}
\lambda_i({\bf u})\neq \lambda_j({\bf u}),\qquad i\neq j,
\end{equation}
at a generic point ${\bf u}$.

\begin{defn}
A semisimple solution $F({\bf v})$ to the WDVV associativity
equations is called \emph{linearly
degenerate} if among the commuting PDEs \eqref{hier1}--\eqref{hier3} there exists at least one
satisfying \eqref{hier4} along with the condition
$$
\frac{\pa \lambda_i({\bf u})}{\pa u_i}=0, \qquad i=1, \dots, n.
$$
\end{defn}

The motivation for our terminology is that one of the quasilinear systems of the commuting family
\eqref{hier1}--\eqref{hier3} is linearly degenerate, i.e., the $i$th characteristic velocity
$\lambda_i$ does not depend on the $i$th Riemann invariant $u_i$
for every $i$ from $i=1$ to $i=n$.

The main goal of the present paper is to classify linearly degenerate
solutions to the WDVV
associativity equations. Such a solution is called \emph{reducible} if, for some $i$, one has
$\gamma_{ij}({\bf u})\equiv 0$ for all $j\neq i$. Otherwise
it will be called
\emph{irreducible}. It suffices to classify irreducible linearly degenerate solutions.

\begin{theorem} \label{mainth}
The rotation coefficients of an irreducible linearly degenerate solution to the WDVV associativity
equations has the form
\begin{equation}\label{glav1}
\gamma_{ij}(u) =\frac{[G (1-\frac1{\rho}\op{tanh}\rho U\cdot G)^{-1}]_{ij}}
{\op{cosh}\rho u_i\op{cosh}
\rho u_j}, \qquad i, j=1, \dots, n, \;i\neq j,
\end{equation}
where $U=\op{diag}(u_1, \dots, u_n)$ and $G$ is a symmetric matrix
satisfying the condition $G^2 =\rho^2\cdot 1$,
in which $\rho$ is an arbitrary complex parameter.
\end{theorem}

For $\rho =0$, the above formulas are considered in the sense of
the limits
$$
\frac1{\rho}\op{tanh} \rho U\to U, \qquad \op{cosh} \rho u_i\to 1.
$$

The paper is organized as follows. In Section 2 we recall the necessary constructions of the theory
of the WDVV associativity equations and derive the basic
system of differential equations
\eqref{main} of the theory of linearly degenerate solutions to the WDVV equations.
In Section 3 we solve the basic
system and describe its symmetry group acting by
fractional linear transformations. In Section 4
we select those solutions to the basic system that give rise to the
WDVV equations and derive the matrix algebraic Riccati
equation. Using the symmetries of this equation, we classify all irreducible
linearly degenerate
solutions to the WDVV associativity equations.

\subhead{Acknowledgments} The authors thank Evgenii Ferapontov and
Sergei Tsarev for stimulating
and clarifying discussions.

\section{Linearly Degenerate Solutions to the WDVV Associativity Equations}

Let $\Gamma=(\gamma_{ij}({\bf u}))_{1\le i,j\le n}$ be the symmetric matrix of rotation
coefficients\ft{Actually, in the differential geometry of curvilinear orthogonal coordinate systems
only the off-diagonal entries of the matrix $\Gamma$ are called rotation coefficients. However, in
our case it will be convenient to add the diagonal entries $\gamma_{ii}=\pa \log h_i/\pa
u_i$.} \eqref{rot1} of a linearly degenerate irreducible solution to the associativity equations.

\begin{lemma}
The matrix-valued function $\Gamma=\Gamma({\bf u})$ satisfies the
differential equations
\begin{equation}\label{delta1}
\frac{\pa \Gamma}{\pa u_k} =\Gamma E_{k} \Gamma +\sigma_k(u_k) E_{k}, \qquad k=1, \dots, n,
\end{equation}
with some functions $\sigma_1(u_1),\dots,\sigma_n(u_n)$. Here $E_{k}$ is a
matrix with only one
nonzero entry, namely,
\begin{equation}\label{unity}
(E_{k})_{ij}=\delta_{ik} \delta_{jk}.
\end{equation}
\end{lemma}

\begin{proof}{Proof}
By construction the equations
\begin{equation}\label{delta2}
\frac{\pa \gamma_{ij}}{\pa u_k} =\gamma_{ik}
\gamma_{kj}
\end{equation}
hold true for distinct values of the indices $i$, $j$, and $k$.
Let us first prove that
\eqref{delta2} also holds when $k=i$ or $k=j$ and
$i\neq j$ or when $i=j$ but $k\neq i$.

According to \cite{npb}, the characteristic velocities $\lambda_k({\bf u})$
of the commuting PDEs
\eqref{hier1}--\eqref{hier3} can be represented in the form
$$
\lambda_k({\bf u}) =\frac{\phi_k({\bf u})}{h_k({\bf u})}, \qquad k=1, \dots, n,
$$
where the vector-function $\phi=(\phi_1({\bf u}),\dots,\phi_n({\bf u}))$ satisfies the system of
linear differential equations
\begin{equation}\label{psi2}
\frac{\pa\phi_i}{\pa u_j} = \gamma_{ij}\phi_j, \qquad i\neq j.
\end{equation}
In particular, $\phi_k=h_k$ is one of the solutions to \eqref{psi2}.
Let $\phi$ be the solution to
\eqref{psi2} corresponding to a linearly degenerate member of the commuting family
\eqref{hier1}--\eqref{hier3}. Differentiating the equation
$$
\frac{\pa}{\pa u_k} \left(\frac{\phi_k}{h_k}\right) =0
$$
in $u_i$ with $i\neq k$, we obtain the
equation
$$
\frac{h_i}{h_k} (\lambda_i-\lambda_k)\,\gamma_{ik}\,\frac{\pa}{\pa u_k}[ \log \gamma_{ik} - \log
h_k]=0.
$$
Due to the assumptions of irreducibility and \eqref{hier4}, we arrive at the equation
$$
\frac{\pa \log\gamma_{ik}}{\pa u_k} =\frac{\pa \log h_k}{\pa u_k}=\gamma_{kk}.
$$
This proves \eqref{delta2} for the case where $k=j$ and $i\neq j$.
Next, assuming that $k\neq i$, one has
$$
\frac{\pa \gamma_{ii}}{\pa u_k} =\frac{\pa}{\pa u_i}\,\frac{\pa \log h_i}{\pa u_k} =\frac{\pa}{\pa
u_i}\left(\gamma_{ik} \frac{h_k}{h_i}\right) =\gamma_{ik}^2.
$$
Thus, Eq. \eqref{delta2} with $i=j$ and $k\neq i$ is also
verified. The last step is to verify
that the difference $\sigma_i:=\pa \gamma_{ii}/\pa u_i -\gamma_{ii}^2$ depends only on $u_i$.
Indeed, for $k\neq i$,
$$
\frac{\pa }{\pa u_k} \left(\frac{\pa \gamma_{ii}}{\pa u_i} -\gamma_{ii}^2\right) =\frac{\pa }{\pa
u_i}\frac{\pa \gamma_{ii}}{\pa u_k} -2 \gamma_{ii} \gamma_{ik}^2 =\frac{\pa\gamma_{ik}^2}{\pa u_i}
-2 \gamma_{ii} \gamma_{ik}^2=0.\eqno{\Box}
$$
\end{proof}

Now, let us describe a class of transformations
$$
u_k\mapsto \tilde u_k, \qquad \gamma_{ij}\mapsto \tilde \gamma_{ij}
$$
which leave system \eqref{delta1}
invariant.

\begin{lemma} \label{lemma22}
The substitution
\begin{equation}\label{trans1}
\begin{alignedat}2
\tilde u_k&= f_k(u_k),&\qquad &k=1, \dots, n,\\
\tilde\gamma_{ij}&= \frac{\gamma_{ij}}{\sqrt{f_i'(u_i) f_j'(u_j)}}-\frac{f_i''(u_i)}{2
[f_i'(u_i)]^2}\,\delta_{ij},&\qquad & i, j=1, \dots, n,
\end{alignedat}
\end{equation}
with arbitrary nonconstant smooth functions $f_1(u_1),\dots,f_n(u_n)$ leaves invariant the form
of\break
Eqs.~\eqref{delta1}, which transform into
$$
\frac{\pa \wt\Gamma}{\pa \tilde u_k} =\wt\Gamma E_{k} \wt\Gamma +\tilde \sigma_k(\tilde u_k) E_{k},
\qquad k=1, \dots, n,
$$
with ${f_k'}^2 \tilde\sigma_k =\sigma_k -\frac12 S_{u_k}(f_k)$. Here $S_u(f)$ is the Schwarzian
derivative of a function $f=f(u)$, that is,
$$
S_u(f) =\frac{f'''}{f'} -\frac32 \frac{{f''}^2}{{f'}^2}\,.
$$
\end{lemma}

This lemma is proved by a straightforward
calculation.\qed

\begin{cor}
A suitable transformation of the
form \eqref{trans1} reduces system \eqref{delta1} to
the form
\begin{equation}\label{main}
\frac{\pa \wt\Gamma}{\pa \tilde u_k} =\wt\Gamma E_{k} \wt\Gamma, \qquad k=1, \dots, n.
\end{equation}
\end{cor}

\begin{proof}{Proof}
The needed transformation $\tilde u_k =f_k(u_k)$ is determined from the Schwarzian equations
$$
S_{u_k} (f_k)=2\sigma_k(u_k), \qquad k=1, \dots, n.\eqno{\Box}
$$
\end{proof}

Recall that the solution to the general Schwarzian equation $S_u (f(u)) =2 \sigma(u)$
can be represented
as the ratio of two solutions to the linear second-order equation
$$
y''+\sigma(u) y =0.
$$

\begin{remark}
System \eqref{main} was studied in \cite{ekpz}
in the investigation of the so-called multi-flow cold gas
reductions of the nonlocal kinetic equation derived as the thermodynamical
limit of the averaged
multi-phase solutions of the KdV equation by the Whitham method.
\end{remark}

In the next section we shall solve system \eqref{main}.

\section{Basic System}

In this section we shall describe solutions to the basic system
\begin{equation}\label{sym11}
\frac{\pa \Gamma}{\pa u_k} =\Gamma E_{k} \Gamma, \qquad k=1, \dots, n.
\end{equation}
Here
$$
\Gamma=(\gamma_{ij}({\bf u}))_{1\le i, j\le n}
$$
is a symmetric matrix (the tildes used in the previous section
are omitted). The compatibility
conditions
$$
\frac{\pa}{\pa u_l}\,\frac{\pa \Gamma}{\pa u_k} = \frac{\pa}{\pa u_k}\,\frac{\pa \Gamma}{\pa u_l}
$$
for any $k$ and $l$ can be readily verified. So, locally, any solution
to \eqref{sym11} is uniquely
determined by the initial data
$$
\Gamma^0=\Gamma({\bf u}^0).
$$
Here ${\bf u}^0$ is any point in the space of independent variables. Therefore, the space of
solutions to the system \eqref{sym11} has dimension $n(n+1)/2$.

Without loss of generality, one can assume that ${\bf u}^0=0$.
The solution to system \eqref{sym11}
with given initial data at the point ${\bf u}=0$ can be written explicitly.

\begin{prop}
The solution $\Gamma=\Gamma({\bf u})$ to the basic system \eqref{sym11}
with initial data
$$
\Gamma(0) =G,
$$
where $G=(g_{ij})$ is a given symmetric matrix, is determined
by the formula
\begin{equation}\label{otvet1}
\Gamma=G (1-U G)^{-1},
\end{equation}
where $1$ is the $n\times n$ identity matrix and $U= \op{diag}(u_1, \dots, u_n)$.
\end{prop}

\begin{proof}{Proof}
The symmetry of the matrix \eqref{otvet1} is tantamount to the
relation
$$
G (1-U G)^{-1} = (1- GU)^{-1} G.
$$
To prove this relation, we multiply it by $1-GU$ on the left and by $1-UG$ on the right
and arrive at the
obvious identity $(1- GU) G= G(1-U G)=G-GU G$. Clearly, $\Gamma(0)=G$.
The proof of the proposition is completed by applying the well-known
rule
$$
\frac{\pa \Gamma}{\pa u_k}=- G (1-U G)^{-1} \frac{\pa (1-U G)}{\pa u_k} (1-U G)^{-1}=G (1-U
G)^{-1} E_{k} G (1-U G)^{-1}=\Gamma E_{k}\Gamma
$$
for differentiating inverse matrices.\qed
\end{proof}

\begin{exam}
For a matrix $g_{ij}=\omega_i \omega_j$ of rank $1$,
one obtains the following solution to the
basic system:
\begin{equation}\label{exam1}
\gamma_{ij} =\frac{\omega_i \omega_j}{1-\sum_{k=1}^n \omega_k^2 u_k}\,.
\end{equation}
\end{exam}

Now, let us describe a subclass of
transformations \eqref{trans1} leaving invariant the basic
system \eqref{sym11}.

\begin{prop}
The basic system \eqref{sym11} is invariant with respect
to transformations \eqref{trans1} if and only if
$f_k(u_k)$ for every $k=1, \dots, n$ is a fractional linear transformation
$$
f_k(u_k)=\frac{a_k u_k +b_k}{c_k u_k + d_k}, \qquad a_kd_k-b_kc_k=1.
$$
\end{prop}

\begin{proof}{Proof}
It is well known that the general solution to the homogeneous Schwarzian equation
$$
\frac{f'''}{f'} -\frac32 \frac{{f''}^2}{{f'}^2}=0
$$
is given by a fractional linear function. \qed
\end{proof}

\begin{cor} The basic system \eqref{sym11} is invariant
with respect to the transformations
\begin{equation}\label{trans3}
\begin{aligned}
\tilde u_k&=\frac{a_k u_k +b_k}{c_k u_k+d_k}, \qquad \begin{pmatrix} a_k & b_k\\ c_k &
d_k\end{pmatrix}\in SL_2(\mathbb R),\; k=1, \dots, n,\\
\tilde \gamma_{ij}&=(c_iu_i+d_i) (c_j u_j+d_j) \gamma_{ij} + c_i (c_iu_i+d_i) \delta_{ij}, \qquad i,
j=1, \dots, n.
\end{aligned}
\end{equation}
\end{cor}

The matrix version of transformation \eqref{trans3} is
\begin{equation}\label{trans32}
\wt{U}=(AU + B)(CU+D)^{-1}, \qquad \wt\Gamma=(CU+D)\Gamma (CU+D)+C(CU+D),
\end{equation}
where $A=\op{diag}(a_1, \dots, a_n)$, $B=\op{diag} (b_1, \dots, b_n)$, $C=\op{diag}(c_1, \dots,
c_n)$, $D=\op{diag}(d_1,\dots, d_n)$, and $AD - BC=1$.

\begin{exam} The substitution
$$
\tilde u_k = \omega_k^2 u_k, \qquad \tilde \gamma_{ij} =\frac{\gamma_{ij}}{\omega_i \omega_j}
$$
reduces solution \eqref{exam1} to the standard form
$$
\tilde\gamma_{ij} =\frac1{1-\sum_{k=1}^n \tilde u_k}, \qquad i, j=1, \dots, n.
$$
\end{exam}

The action of the $[SL_2(\mathbb R)]^n$ transformations \eqref{trans3} on
solutions \eqref{otvet1}
is given by the following analogue of Siegel modular transformations.

\begin{prop} Let the symmetric matrix $G$ satisfy the condition
$\det (A+BG)\neq 0$. Then transformation \eqref{trans3}
transforms the solution $\Gamma({\bf
u})$ with initial data $\Gamma(0)=G$ into
$$
\wt\Gamma = \wt{G} (1- \wt{U} \wt{G})^{-1}
$$
with
\begin{equation}\label{trans31}
\wt{G} = (C+DG) (A + BG)^{-1}.
\end{equation}
\end{prop}

\begin{proof}{Proof}
An easy calculation employing \eqref{trans32} yields
$$
\wt\Gamma= (-C\wt{U} +A)^{-1} G [A+BG -\wt{U} (C+DG)]^{-1}+C (-C\wt{U} +A)^{-1}.
$$
Computing the initial data of this solution at $\tilde {\bf u}=0$,
we arrive at $\wt\Gamma(0)=\wt{G}$
with the matrix $\tilde G$ given by \eqref{trans31}. \qed
\end{proof}

\begin{defn}
Two solutions $\Gamma$ and $\tilde \Gamma$ to the basic system
are called \emph{equivalent} if they
are related by a symmetry transformation of the form \eqref{trans32}.
Two symmetric matrices $G$ and $\wt{G}$
related by transformation \eqref{trans31} will also be called equivalent.
\end{defn}

Note that the useful identity
\begin{equation}\label{trans33}
(C+DG) (A + BG)^{-1}=(A+GB)^{-1} (C+GD)
\end{equation}
is equivalent to the symmetry of the matrix $G$.

\section[Solutions of the Associativity Equations]{From Solutions of the Basic System
to Linearly Degenerate Solutions of the Associativity Equations}

In this section we address the problem of selecting
those solutions to the basic system
\eqref{sym11} that come from a linearly degenerate solution to the associativity equations.

Given a symmetric matrix-valued function
$\Gamma({\bf u})$ satisfying \eqref{sym11}, we
look for a substitution of the form \eqref{trans1} such that
the transformed matrix $\wt\Gamma$
satisfies also the last equation \eqref{egor11} of the Darboux--Egorov system, that is,
\begin{equation}\label{eg1}
\sum_{k=1}^n \frac{\pa \wt\Gamma}{\pa\tilde u_k}\quad\text{is a diagonal matrix}.
\end{equation}
Recall that the equations
$$
\frac{\pa \tilde \gamma_{ij}}{\pa \tilde u_k}=\tilde \gamma_{ik}\tilde \gamma_{kj}\quad
\text{for distinct
$i$, $j$, and $k$,}
$$
which are the first part of this system (Eqs.~(1.4)),
follow from the basic system by Lemma \ref{lemma22}.

Applying Lemma \ref{lemma22}, we arrive at the following simple statement.

\begin{prop}
Let $\Gamma({\bf u})$ be a solution to the basic system \eqref{sym11}.
Suppose that the functions
$f_1(u_1),\dots,f_n(u_n)$ are chosen in such a way that the transformed matrix \eqref{trans1}
satisfies \eqref{eg1}. Then the off-diagonal entries of the transformed matrix $\wt\Gamma$ are
the rotation coefficients of some Egorov metric.
\end{prop}

We introduce the diagonal matrices
\begin{alignat}2
S&=\op{diag} (s_1, \dots, s_n),&\qquad s_i&=\frac1{f_i'},\\
S'&=\op{diag} (s_1', \dots, s_n'),&\qquad s_i'&=\frac{ds_i}{du_i}=-\frac{f_i''}{[f_i']^2}\,.
\end{alignat}
Here and in the sequel we use the short notation
$$
f_i'=f_i'(u_i), \quad f_i''=f_i''(u_i),\quad\text{etc.}
$$
In this notation the transformation law \eqref{trans1} reads
$$
\wt\Gamma=S^{1/2} \Gamma S^{1/2} +\tfrac12 S'.
$$
Thus, condition \eqref{eg1} can be represented in the form
\begin{equation}\label{eg3}
\Gamma S \Gamma +\tfrac12 S'\Gamma +\tfrac12 \Gamma S' +P=0
\end{equation}
for some diagonal matrix $P$.

\begin{defn} A solution $\Gamma$ is called \emph{reducible} if,
for some $i$,
$$
\gamma_{ij}\equiv 0\quad\text{for any }\,j\neq i.
$$
Otherwise it is called \emph{irreducible}.
\end{defn}

A reducible solution is expressed in terms of functions depending
on a smaller number of variables.

\begin{theorem} For an irreducible solution
$$
\Gamma=G(1-UG)^{-1}=(1-GU)^{-1}G,
$$
a transformation of the form \eqref{trans1} satisfying \eqref{eg1}
exists if and only if the matrix $G$ satisfies
the quadratic equation
\begin{equation}\label{quad1}
GRG+QG+GQ+P=0
\end{equation}
for some constant diagonal matrices
$$
P=\op{diag}(p_1, \dots, p_n), \quad Q=\op{diag}(q_1, \dots, q_n), \quad R=\op{diag}(r_1, \dots,
r_n).
$$
The transformation in question is determined by
$$
\frac{d \tilde u_i}{d u_i} =\frac1{p_i u_i^2 + 2 q_i u_i + r_i}, \qquad i=1, \dots, n.
$$
\end{theorem}

\begin{proof}{Proof}
Differentiating \eqref{eg3} in $u_i$ and using \eqref{sym11}
and the obvious formulas
$$
\frac{\pa S}{\pa u_i}=s_i' E_{i}, \qquad \frac{\pa S'}{\pa u_i}=s_i''E_{i},
$$
etc., one obtains
\begin{equation}\label{eg5}
\left(\frac12 s_i''-p_i\right) (\Gamma
E_{i}+E_{i}\Gamma) +\frac{\pa P}{\pa u_i}=0.
\end{equation}
All entries of the matrix $\Gamma E_{i}+E_{i}\Gamma$ vanish, except
the $i$th
row and the $i$th
column, which coincide with $(\gamma_{1i}, \dots, \gamma_{ni})$.
Due to the irreducibility assumption, it follows
from \eqref{eg5} that
\begin{equation}\label{eg6}
p_i=\frac12\,s_i''.
\end{equation}
Substituting this into \eqref{eg5} yields
$$
\frac{\pa P}{\pa u_i}=0.
$$
Repeating this procedure for every $i=1, \dots, n$, one proves that
the matrix $P$ is constant. Using
\eqref{eg6}, we conclude that $s_i=s_i(u_i)$ is a quadratic polynomial, i.e.,
$s_i= p_i u_i^2+ 2 q_i u_i +
r_i$. Finally, multiplying Eq.~\eqref{eg3} by $1-GU$ on the left
and by $1-UG$ on the right, we
arrive at the quadratic equation \eqref{quad1}. \qed
\end{proof}

\begin{defn}
A symmetric matrix $G$ is called \emph{admissible} if it satisfies the matrix quadratic equation
\eqref{quad1}. A solution of the form $\Gamma=G(1-UG)^{-1}$ is called
\textit{admissible} if the parameter matrix $G$ is admissible.
\end{defn}

The matrix quadratic equation \eqref{quad1} for the symmetric matrix $G$ is a particular case
of the so-called \emph{algebraic Riccati equation} (see, e.g., \cite{lancaster}). The class of such
equations is invariant with respect to fractional linear transformations,
as the following lemma shows.

\begin{lemma}
If a symmetric matrix $G$ satisfies the matrix quadratic equation
$$
GRG +QG +GQ +P=0
$$
with some diagonal matrices $P$, $Q$, and $R$, then
the equivalent matrix $\wt{G}=(C+DG)
(A+BG)^{-1}$ satisfies an equation of the same form
$$
\wt{G}\wt{R}\wt{G} +\wt{Q}\wt{G} +\wt{G}\wt{Q} +\wt{P}=0
$$
with
\begin{equation}\label{quad22}
\begin{aligned}
\wt{P}&= D^2 P -2 CDQ +C^2 R,\\
\wt{Q}&= -BDP +(AD+BC)Q-ACR,\\
\wt{R}&= B^2 P -2 ABQ +A^2 R.
\end{aligned}
\end{equation}
\end{lemma}

The proof of this lemma is straightforward
and uses identity \eqref{trans33}.\qed

\begin{cor}
The class of admissible solutions to the basic system \eqref{sym11}
is invariant with respect to the
$[SL_2]^n$ action \eqref{trans32}.
\end{cor}

The entries $\Delta_1,\dots,\Delta_n$ of the diagonal matrix
\begin{equation}\label{delta0}
\Delta=Q^2-PR
\end{equation}
are invariants of the $[SL_2]^n$ action \eqref{quad22}.

The next step is to parameterize linearly degenerate solutions to the associativity equations by
solutions to the algebraic Riccati equation \eqref{quad1} with prescribed coefficients
satisfying the condition
$$
|p_i|^2+|q_i|^2+|r_i|^2\neq 0, \qquad i=1, \dots, n.
$$
Let us first simplify the matrix quadratic equation by means of transformations
\eqref{quad22}.

\begin{lemma}\label{Lem47} (1) For an irreducible admissible matrix $G$,
the matrix quadratic equation
\eqref{quad1} is equivalent, up to transformations \eqref{quad22}, to
\begin{equation}\label{canon1}
G^2=\Delta,
\end{equation}
where $\Delta$ is given by \eqref{delta0}.

(2) For an admissible irreducible $G$, the matrix $\Delta$ is proportional
to the identity matrix, i.e.,
$$
\Delta_1=\dots =\Delta_n=:\rho^2.
$$
\end{lemma}

\begin{proof}{Proof}
If all entries of the matrix $R$ are different from zero, then
Eq.~\eqref{quad1} can be
reduced to the canonical form \eqref{canon1} by a transformation of the form
$$
G\mapsto AGA+B
$$
with suitable diagonal matrices $A$ and $B$. This is a particular class of transformation
\eqref{quad22}. If $r_i=0$ for some $i$, then one can assume that $p_i\neq 0$. Let us apply
the fractional linear transformation of the form \eqref{trans31} with $A=1-E_{i}$, $B=-E_{i}$,
$C=E_{i}$, and $D=1-E_{i}$, that is, $G\mapsto \wt{G}=[G+E_i (1-G)] [
1-E_{i}(1+G)]^{-1}$, where the matrix $E_{i}$ is of the form \eqref{unity}. Such a transformation is
applicable only if the matrix $1-E_{i}(1+G)$ is nonsingular. It is easy to see that the determinant of this matrix is equal to $\pm
g_{ii}=\gamma_{ii}(0)$. If $g_{ii}=0$ but the solution is irreducible, then one can perform a
shift ${\bf u}\mapsto {\bf u}+{\bf u}^0$ to obtain a matrix $G'=\Gamma({\bf
u}^0)$ with $g'_{ii}\neq 0$. After the transformation, one obtains $\tilde r_i=p_i\neq 0$.

To prove the second part of the lemma,
it suffices to observe that any eigenvector $f$ of
the matrix $G$ with eigenvalue $\lambda$ is an eigenvector
of $G^2$ with eigenvalue
$\lambda^2$. So, if $e_i$ and $e_j$ are the $i$th and $j$th basic
vectors and $\Delta_i\neq
\Delta_j$, then these vectors belong, respectively,
to the sums of root subspaces
$R(\sqrt{\Delta_i})\oplus R(-\sqrt{\Delta_i})$
and $R(\sqrt{\Delta_j})\oplus R(-\sqrt{\Delta_j})$ of
the matrix $G$. Such root subspaces of symmetric matrices
are orthogonal; hence the matrix $G$ must
have block-diagonal form in the same basis.\qed
\end{proof}

The main Theorem \ref{mainth} readily follows from the above considerations.

Recall that the reconstruction of the solution to the associativity equations with given rotation
coefficients \eqref{glav1} depends on the choice of a solution to the linear system \eqref{egor2},
\eqref{egor3}. Below we apply this procedure to produce examples of linearly degenerate
 WDVV solutions. It is convenient to separately consider the cases $\rho\neq 0$ and $\rho=0$.

\textbf{Case 1.} The eigenvalues of a symmetric matrix $G$
satisfying $G^2 =\rho^2\cdot 1$ are equal to
$\pm \rho$. Let $k$ denote the number of eigenvalues
equal to $-\rho$. We consider the case $k=1$ in more
detail. It is more convenient to deal with the matrix
$\wt{G}=G-\rho\cdot 1$, which
satisfies the equation $\wt{G}^2 +2\rho\wt{G}=0$. In
the case $k=1$, this
matrix can be represented in the form
$$
\wt{G}=(\omega_i \omega_j), \qquad \sum_{i=1}^n \omega_i^2=-2\rho.
$$
To this matrix there corresponds a family of solutions of the
form \eqref{exam1}.
The substitution $\tilde
u_k = -\log [ \omega_k^2 (u_k-u_k^0)]$, $k=1, \dots, n$, with arbitrary
constants $u_k^0$ satisfying
$\sum_{k=1}^n u_k^0=0$ yields the following rotation coefficients satisfying the Darboux--Egorov
equations:
$$
\tilde\gamma_{ij}=\frac{e^{-(\tilde u_i +\tilde u_j)/2}}{\sum_{k=1}^n e^{-\tilde u_k}}, \qquad i\neq
j.
$$
In the sequel, we omit the tildes.
System \eqref{egor2}--\eqref{egor3} can be easily solved:
\begin{gather*}
\psi_{ii}=\frac{2e^{-u_i}}{D}-1,\quad \psi_{ij}=\frac{2e^{-(u_i + u_j)/2}}{D},\quad i\neq
j,\quad\text{where }\, D=\sum_{k=1}^n e^{-u_k}.
\end{gather*}
The calculation of the quadratures \eqref{egor4}
gives the following expression for the matrix $\Omega$
of the second derivatives of the potential (see \eqref{egor41}):
\begin{equation}\label{rank3}
\Omega_{ij}=u_i \delta_{ij}+\frac{4 e^{-(u_i+u_j)/2}}{D}\,.
\end{equation}

Flat coordinates are obtained by choosing a linear combination of the
columns of this matrix. The
choice of the first column yields the Egorov metric
$$
ds^2 =\left(1 - 4\,\frac{e^{-u_1}}{D}\right) du_1^2 +4 \sum_{i=1}^n
\frac{e^{-u_1-u_i}}{D^2}\,du_i^2
$$
with the flat coordinates
$$
v_1=u_1 +\frac{4 e^{-u_1}}{D}, \qquad v_i = \frac{4 e^{-(u_1+u_i)/2}}{D}\quad\text{for }\,i\neq 1.
$$
Solving these equations for the canonical coordinates $u_i$, we obtain
$$
u_1=v_1-\sqrt{4-\sigma} -2, \qquad u_i=v_1 -\sqrt{4-\sigma}-2 +2\log\frac{2+\sqrt{4-\sigma}}{v_i}
\quad\text{for }\,i\neq 1
$$
with $\sigma=\sum_{k=2}^n v_k^2$, and integrating
quadratures \eqref{rank3}, we arrive at the
following expression for the potential being the corresponding linearly
degenerate solution to the WDVV
associativity equations:
\begin{equation}\label{konec}
F=\frac16\,v_1^3 +\frac12\,v_1 \sigma -\sum_{k=2}^n v_k^2 \log v_k -\frac13\,
(2+\sigma)\sqrt{4-\sigma} +\sigma\log(2+\sqrt{4-\sigma}).
\end{equation}

One can also obtain an explicit realization of the integrable hierarchy associated, in the sense of
\cite{npb}, with \eqref{konec}. Recall that the hierarchy is
an infinite family of commuting flows labeled by pairs $(\alpha, p)$, $\alpha=1, \dots, n$, $p=0,
1,2, \dots$. The flows have the form
$$
\frac{\pa v^\gamma}{\pa t^{\alpha,p}} = \partial_x (\nabla^\gamma \theta_{\alpha,p+1}(v)).
$$
The generating functions
$$
\theta_\alpha(v,z)=\sum_{p=0}^\infty \theta_{\alpha, p}(v) z^p
$$
of $\theta_{\alpha,p}(v)$ (\emph{deformed flat coordinates}) can be found
in quadratures; we have
$$
d\theta_\alpha(v,z) =\sum_{i=1}^n h_i \Psi_{i\, \alpha} du_i, \qquad \alpha=1, \dots, n,
$$
where the $\Psi_{i\alpha}(v,z)$, $\alpha=1, \dots, n$,
form a basis for the ``wave functions" determined by the
system
\begin{align*}
\frac{\pa \Psi_i}{\pa u_j}&=\gamma_{ij} \Psi_j, \qquad i\neq j,\\
\sum_{k=1}^n \frac{\pa \Psi_i}{\pa u_k}&=z\Psi_i.
\end{align*}
The basis $\Psi_{i\alpha}$ can be conveniently orthonormalized by the conditions
$$
\sum_{\alpha=1}^n \Psi_{i\alpha}(v, -z)
\Psi_{j\alpha}(v,z)=\delta_{ij}.
$$
In our case the normalized wave functions have the form
$$
\Psi_{i\alpha}=\frac{2 e^{z\, u_\alpha}}{\sqrt{1-4 z^2}}\bigg[\bigg(z-\frac12\bigg)
\delta_{i\alpha} +\frac{e^{-\frac{u_i+u_\alpha}2}}{D}\bigg].
$$
This gives
$$
\theta_\alpha=\frac1{\sqrt{1-4 z^2}} \bigg\{\bigg[ \frac1{z}(e^{zu_1} -1) -e^{zu_1} (u_1+2)+2
\bigg] \delta_{\alpha 1} +v_\alpha e^{zu_\alpha}\bigg\}, \qquad \alpha=1, \dots, n.
$$

\textbf{Case 2.} Now, consider the second type of
solutions, namely, those parametrized by symmetric matrices
$G$ satisfying $G^2=0$. In this case, one again obtains a solution
to the WDVV equations
which satisfies
the quasihomogeneity condition.

All eigenvalues of $G$ are equal to $0$. All Jordan blocks are of order $1$ or
$2$. Consider the simplest case of only one block of order $2$. The entries
of the matrix $G=(g_{ij})$ can be written in the form
$$
g_{ij}= \omega_i \omega_j, \qquad \sum_{i=1}^n \omega_i^2=0.
$$
The corresponding solution to the WDVV system can be obtained
from the trivial (i.e., cubic) solution
$$
F(v)=\frac16 \sum_{i,j,k} c_{ijk} v^i v^j v^k
$$
by applying the inversion symmetry described in \cite{monte}
(see Appendix B and Proposition 3.14 in~\cite{monte}).
Here the $c_{ijk}$
are the structure constants of the semisimple Frobenius algebra
$$
\mathcal{A}=\op{span}(e_1, \dots, e_n), \quad \langle e_i\cdot e_j, e_k\rangle =c_{ijk}, \quad
\langle e_i, e_j\rangle=\delta_{i+j, n+1}
$$
with a unit $e_1$ and trivial grading
$\op{deg}e_i =0$ for all $i$. Recall that the
structure constants can be represented in the form
$$
c_{ijk}=\sum_{s=1}^n \frac{a_{si} a_{sj} a_{sk}}{a_{s1}},
$$
where the matrix $(a_{ij})$ satisfies the condition
$$
\sum_{s=1}^n a_{si} a_{sj}=\delta_{i+j,n+1}.
$$
For our construction, we can choose the matrix in such a way that
$$
a_{i1}=\omega_i, \qquad i=1, \dots, n.
$$
After the substitution of
\begin{align*}
\hat v^1&=\frac12\,\frac{v_\alpha v^\alpha}{v^n},\\
\hat v^\alpha&=\frac{v^\alpha}{v^n},\qquad \alpha\neq 1,n,\\
\hat v^n&=-\frac1{v^n}
\end{align*}
one obtains the needed solution $\wh{F}$ to the WDVV equations in the form
\begin{equation}\label{triv5}
\wh{F}(\hat v) =\frac12\,\hat v^1 \hat v_\alpha \hat v^\alpha+(\hat v^n)^{2} F(v)=\frac12 (\hat
v^1)^2 \hat v^n +\frac12 \sum_{\alpha=2}^{n-1} \hat v^1\hat v^\alpha \hat v^{n-\alpha+1}+\frac{P(\hat v^2, \dots, \hat v^{n-1})}{\hat v^n}\,.
\end{equation}
Here\ft{This example was considered in \cite{leur}
in a different context. Our formula
\eqref{triv5} differs from that given in
\cite{leur}.} $P(\hat v^2, \dots, \hat v^{n-1})$ is
a certain
polynomial of degree 4. The potential $\wh{F}$ satisfies the quasihomogeneity condition
$$
\wh{E}\wh{F} = \wh{F}, \qquad \wh{E} = \hat v^1 \frac{\pa}{\pa \hat v^1} -\hat v^n \frac{\pa}{\pa
\hat v^n}\,.
$$

\vskip10pt

\end{document}